\documentclass[a4paper,12pt]{article}
\usepackage{amsfonts}
\usepackage{latexsym}
\usepackage{amsmath}
\usepackage{MnSymbol}
\usepackage{indentfirst}
\newtheorem{theorem}{Theorem}[section]

\begingroup
\newtheorem{definition}[theorem]{Definition}

\newtheorem{proposition}[theorem]{Proposition}

\newtheorem{remark}[theorem]{Remark}
\endgroup

\newenvironment{proof} {\textbf{Proof.}} {\ }

\begin{document}

\noindent{{\Large{\textbf{Algebraic and relational models for a system \\ 
based on a poset of two elements}}} 
\footnote{Research realized in the framework of COST Action n$^\circ$ 15 (in Informatics) ``Many-Valued Logics for Computer Science Applications''}
}

\

\noindent \small{LUISA ITURRIOZ,
\textit{Universit\'e de Lyon, Universit\'e Claude Bernard Lyon 1, \\
Institut Camille Jordan, CNRS UMR 5208,
F-69622 Villeurbanne cedex, France.} \\
\noindent E-mail: luisa.iturrioz@math.univ-lyon1.fr}

\

\noindent \textbf{Abstract:}
The aim of this paper is to present a very simple set of conditions, necessary 
for the management of knowledge of a poset $T$ of two agents, which are 
partially ordered by the capabilities available in the system. We build up a formal system and we elaborate suitable semantic models in order to derive information from the poset.
The system is related to three-valued Heyting algebras with Boolean operators.

\

\noindent \textbf{Key Words:}
Distributive lattices with Boolean operators, $T$-structures, three-valued Heyting algebras, algebraic and relational models, knowledge representation

\section{\hspace{-2.6ex}.{\hspace{1.ex}}Introduction}

\noindent The purpose of this paper is to provide a propositional logical framework for representing and reasoning about knowledge of a poset of two agents (e.g. in robotics). The situation we have in mind may be described as follows. 
 
Assume $T$ is a poset of two agents $t_1$ and $t_2$. 
We denote $t_1 \leq t_2$ to express the fact that agent $t_2$ has more 
possibilities than agent $t_1$. In the applications, $T$ may be considered 
to be a poset of two co-operating intelligent agents partially ordered by the competences about a particular domain, as for example a ``knower" and a ``learner".

We suppose that a minimal necessary ingredient of a formal system that 
is capable of simulating a practical reasoning must include a lattice 
structure to manage the connectives ``and" and ``or". 

For agent $t_i$, the intuitive meaning of the connective $S_{t_i} a$ is: ``agent $t_i$ perceives the information $a$". Related to the lattice structure, perception operators are asked to be compositional.

\

Mathematical simple structures that we explore in modelling our ideas may be presented in the following way.

On a distributive lattice $(A, 0, 1, \wedge,\vee)$ with 
zero and unit we are going to define three unary operators, denoted $C, S_{t_1},S_{t_2}$. 
Perception operators $S_{t_1},S_{t_2}$ are asked to be compositional, 
Boolean, and accepting individual opinions without any change; 
$C$ which is considered here only to give a neat definition below, 
is understood to satisfy the equalities:  
$S_{t_1} a \wedge Ca = 0$ and $S_{t_1} a \vee Ca = 1$, for all $a \in A$.

\

Thus, the required properties for these operators are the following, for all $a \in A$: 
\begin{itemize}
\item the operators $S_t$, for $t \in \{ t_1,t_2 \}$, are $(0, 1)$-lattice homomorphisms 
from $A$ onto the sublattice $B(A)$ of all complemented elements of $A$ such that $S_t S_w a = S_w a$, for all $t, w \in \{ t_1,t_2 \}$,
\item $S_{t_1} a \leq S_{t_2} a$,
\item $S_{t_1}$ is related to the operation C by the equations: \\ 
$S_{t_1} a \wedge Ca = 0$ and $S_{t_1} a \vee Ca = 1$.  
\end{itemize}

We remark that, for arbitrary elements $a, b \in A$, the relation $``\equiv"$ defined in the following way: 
\[
a \equiv b \quad \text{if and only if}\quad S_t a = S_t b, \quad \text{for} \ t \in \{ t_1,t_2 \}.
\]

\noindent is an equivalence relation on A. With respect to the connectives 
$\wedge, \vee, C, S_{t_1}, S_{t_2}$ it is an easy calculation to check that 
it is a congruence in $A$.

In view of this fact, we can identify elements in $A$ if and only if 
agents in $T$ have the same insights on them.

\

The paper is organised as follows. In Section 2, the definition of the algebraic structure is derived and a fundamental example is exhibited. Other examples are in \cite{Itu01} and \cite{Itu07}.
As the definition is not suitable for logic considerations, we give an equational definition in Section 3. In Section 4 a formalized propositional language is introduced as well as two adapted semantics. The equivalence of algebraic and relational semantics is shown in Section 5. Finally, in Section 6, the question of the decidability of the system is answered.

\section{\hspace{-2.6ex}.{\hspace{1.ex}}An algebraic structure and a fundamental example}

All the above constraints suggest to consider a three-valued structure that we have studied in \cite{Itu01} and \cite{Itu07}. 

This structure emerged from a fundamental example presented later and is 
related to ideas of Moisil \cite{Moi40}, \cite{Moi72}, \cite{BFGR}.

\

For notational convenience, we sometimes replace $t_1$ and $t_2$ by their indices (i.e.\ one and two).

\begin{definition} \label{def:alg2} 
An abstract algebra $(A, 0, 1,\wedge,\vee,C, S_1,S_2)$
where $0, 1$ are constants, $C,S_1,S_2$ are unary operations and $\wedge,\vee$ are binary operations is said to be a \textbf{Distributive lattice with three unary operators} if  
\begin{itemize}
\item[] (T1) $(A,0,1,\wedge,\vee)$ is a distributive lattice with zero and unit,

and for every $a,b \in A$ and for all $i,j = 1,2$, the following equations hold:
\item[] (T2) $S_i (a \wedge b) =  S_i a \wedge S_i b$ ; $S_i (a \vee b) =  S_i a \vee S_i b,$ 
\item[] (T3) $S_1 a \wedge Ca = 0$ ; $S_1 a \vee Ca = 1,$
\item[] (T4) $S_i (S_j) a = S_j a,$
\item[] (T5) $S_1 0 = 0$ ; $S_1 1 = 1,$
\item[] (T6) If $S_i a = S_i b$, for all $i= 1,2,$ then $a=b,$ (\it Determination Principle)
\item[] (T7) $S_1 a \leq S_2 a.$
\end{itemize}
\end{definition}

We will refer to a \textbf{$\boldsymbol{T}$-structure} $A$, for short (as in \cite{Itu01} and \cite{Itu07}).
We remark that this definition is not equational and this fact makes it awkward for us.

\begin{proposition}
The following properties are true in any $T$-structure:
\begin{itemize}
\item[] (T8) $S_2 0 = 0$ ; $S_2 1 = 1,$
\item[] (T9) $a \leq b$ if and only if $S_i a \leq S_i b,$ for $i=1,2$, 
\item[] (T10) $S_1 a \leq a \leq S_2 a,$
\item[] (T11) $S_i a \wedge CS_i a = 0$ ; $S_i a \vee CS_i a = 1,$ for $i=1,2$.
\end{itemize}
\end{proposition}
\begin{proof} See \cite{Itu01}.
\end{proof}

\begin{remark}
Let $\boldsymbol{B(A)}$ be the Boolean algebra of all complemented elements in $A$ 
and $\boldsymbol{S_i(A)} = \{x \in A: S_i x = x\}$.  

From \cite{Itu01}, \cite{Itu82} it is well known  that for all $i = 1, 2$,
$\boldsymbol{S_i(A) = B(A)}$. Also, if $``\neg"$ denotes the Boolean negation we have $\neg S_i a = CS_i a$.
\end{remark}

\noindent{\textbf{A fundamental example}}

For the sake of illustration let us consider a very simple example 
depicting the introduced notions.

\medskip

Let  $T = \{t_1, t_2\}$ be an ordered set such that $t_1 \leq t_2$. For each $t \in T$
we denote $F(t)$ the \textbf{increasing subset} of $T$ , i.e.\ 
\[
F(t) = \{w \in T : t \leq w \}.
\]

Let $A$ be the class of the empty set and all increasing sets, i.e.\
\[
A = \{ \emptyset, F(t_2), F(t_1) \}.
\]
 
The class $A$, ordered by inclusion, is an ordered set with three or two elements, and the system $(A, \emptyset, T, \cap, \cup)$, closed under the operations of intersection and union, is a distributive lattice with zero and unit. For each $t \in T$ we define a special operator $S_t$ on $A$ in the following way, for all $X \subseteq A$:
\[
S_t(X) =
 \left\{
 \begin{array}{rl}
T & \text{if} \quad t \in X\\
\emptyset & \text{otherwise.}
\end{array} 
\right.
\]

Finally we define $CX = \neg S_{t_1}(X)$. Thus the system $(A, \emptyset, T, \cap, \cup,C, S_{t_1},S_{t_2})$ is a $T$-structure, called \textbf {basic} $T$-structure and denoted $\boldsymbol{BT}$ or $\boldsymbol{B}$ if it has three or two elements, respectively. Note that $\boldsymbol{B}$ is a subalgebra of $\boldsymbol{BT}$.

For further examples see \cite{Itu01} and \cite{Itu07}.
 
\section{\hspace{-2.6ex}.{\hspace{1.ex}}An equational definition}
 
In order to develop a logic system of any kind, it is convenient to remember (see for example \cite{Ras74}, page 167) that ``implication" seems to be the most important connective. This fact suggests what we do here. 

In \cite{Itu01} we have introduced an \textbf{equational} definition of a $T$-structure by means of a particular intuitionistic implication. 
\begin{definition} \label{def:alg3} 
A \textbf{Heyting algebra with three unary operators} (or \textbf{$\boldsymbol{HT}$-algebra} for short) is an abstract system $A = (A, 0, 1,\wedge,\vee,\Rightarrow,\neg,S_1,S_2)$
such that $0, 1$ are constants, $\neg,S_1,S_2$ are unary operations and $\wedge,\vee,\Rightarrow$ are binary operations satisfying the following conditions, for all $a, b, c \in A:$  
\begin{itemize}
\item[] (HT1) $(A,0,1,\wedge,\vee,\Rightarrow, \neg)$ is a Heyting algebra,

and for every $a,b \in A$ and for all $i,j = 1,2$ the following equations hold:
\item[] (HT2) $S_i (a \wedge b) =  S_i a \wedge S_i b$ ; $S_i (a \vee b) =  S_i a \vee S_i b,$
\item[] (HT3) $S_2 (a \Rightarrow b) = (S_2 a \Rightarrow S_2 b),$
\item[] (HT4) $S_1 (a \Rightarrow b) = (S_1 a \Rightarrow S_1 b) \wedge (S_2 a \Rightarrow S_2 b),$
\item[] (HT5) $S_i S_j a = S_j a,$
\item[] (HT6) $S_1 a \vee a = a,$
\item[] (HT7) $S_1 a \vee \neg S_1 a = 1,\quad \text{with}\quad \neg a = a\Rightarrow 0.$
\end{itemize}
\end{definition}

The next two theorems state the equivalence between the notion of 
$T$-structure and that of $HT$-algebra and are proved in \cite{Itu01}.

\begin{theorem} \label{theo:impl}
Let $(A, 0, 1,\wedge,\vee, C, S_1,S_2)$ be a $T$-structure and 
$\Rightarrow$ and $\neg$ be two operations defined by means of the 
following equations, for all $a, b \in A$:
\begin{align}
a \Rightarrow b &= b \vee \bigwedge_{k=1}^{2} (CS_k a \vee S_k b), \\
\neg a &= a \Rightarrow 0.
\end{align}
Then the algebra $A = (A, 0, 1,\wedge,\vee,\Rightarrow,\neg,S_1,S_2)$ 
is a $HT$-algebra.
\end{theorem}

Conversely:
\begin{theorem}
Let $A = (A, 0, 1,\wedge,\vee,\Rightarrow,\neg,S_1,S_2)$ be a 
$HT$-algebra and let us introduce a new operation $C$ by means 
of the following equation, for all $a \in A:$
\[
Ca = \neg S_1 a
\]
Then the abstract algebra $(A, 0, 1,\wedge,\vee, C, S_1,S_2)$ is a $T$-structure.
\end{theorem}

\begin{remark}
Every $HT$-algebra satisfies the Ivo Thomas axiom \cite{IvoT},
for all $a, b, c \in A$:
\[
((a \Rightarrow c) \Rightarrow b) \Rightarrow (((b \Rightarrow a) 
\Rightarrow b) \Rightarrow b) = 1. 
\]
This equality implies that every $HT$-algebra $A$ is a three-valued Heyting algebra \cite{Mont-L64}.
\end{remark}

\section{\hspace{-2.6ex}.{\hspace{1.ex}}A formalized propositional language}

The logic considered in the following sections is intended to provide 
a framework to manage a poset of two intelligent agents.

A formal system needs a language. In the applications, this language will 
be used as a tool to represent knowledge.
For notational convenience, we use the same symbols for connectives in the language and operations in algebraic structures.

\

The language of $HT$-logics is a propositional language whose formulas are built from propositional variables taken from a countable set $\textbf{VarProp}$ with signs of conjunction $(\wedge)$, disjunction $(\vee)$, implication $(\Rightarrow)$, negation $(\neg)$, and the family $\{ S_{t_1}, S_{t_2} \}$ of unary connectives. Implication $(\Rightarrow)$ and negation $(\neg)$ are intuitionistic connectives, and operators $S_{t_1}$ and $S_{t_2}$ are Boolean operators.

\

The set $\textbf{For}$ of formulas of the logic is the least set 
satisfying the conditions:
\begin{itemize}
\item[-] $\textbf{VarProp}\ \subseteq \textbf{For}$,
\item[-] if $\alpha, \beta \in \textbf{For}$, then $\alpha \wedge \beta, \alpha \vee \beta, \alpha \Rightarrow \beta \in \textbf{For}$,
\item[-] if $\alpha \in \textbf{For}$ then $\neg \alpha, S_{t_1} \alpha, S_{t_2} \alpha \in \textbf{For}$
\end{itemize}

\noindent{\textbf{Semantics of the language}}

In order to formally reason about knowledge, we need a suitable semantic model.
We define a meaning of formulas of the given language by means of notions of model and satisfiability of formulas in a model, in a standard way.

\

\noindent{\textbf{a- Algebraic models}}

Let $\textbf{For}$ be the set of formulas and $A$ a $HT$-algebra. In the set of formulas, the connectives $(\wedge, \vee, \rightarrow, \neg, S_{t_1}, S_{t_2})$ are regarded as algebraic operations. 

A map $h : \textbf{For} \rightarrow A$ is called a \textbf{homomorphism} provided it preserves all the operations on $\textbf{For}$. 

\smallskip

\begin{definition}
An \textbf{algebraic model} for the set of formulas \textbf{For}, is a system $(A, h)$ such that $A$ is a $HT$-algebra and $h: \textbf{For} \rightarrow A$ is a homomorphism.
\end{definition}

A formula $\alpha$ is \textbf{algebraically true in the algebraic model} 
$(A, h)$ iff $h(\alpha) = 1$, and $\alpha$ is \textbf{algebraically valid} (denoted $\models_{Alg} \alpha$) iff $\alpha$ is algebraically true in every algebraic model. 

A formula $\alpha$ is an \textbf{algebraic consequence} of a set of formulas $\Gamma$ in the algebraic model $(A,h)$ (denoted by $\Gamma \models_{A} \alpha)$ iff whenever all the formulas from $\Gamma$ are algebraic true in $(A, h)$, we have $\alpha$ is algebraic true in $(A, h)$;
and $\alpha$ is an \textbf{algebraic consequence} of a set of formulas $\Gamma$ (denoted by $\Gamma \models_{Alg} \alpha)$ iff for every algebraic model $(A, h)$, we have $\Gamma \models_{A} \alpha$.

\

\noindent{\textbf {b- Relational models}} 

Motivated by some results in [\cite{IO-96}, p.135], we introduce the following notion.

\begin{definition}\label{def:HT-frame}
A \textbf{$\boldsymbol{HT}$-frame} is a system 
\[
K = (W, R, s_1, s_2)
\]

\noindent where, for all $w \in W$

\begin{itemize}
\item[(K0)] $W$ is a nonempty set (of states $w$), $R$ is a binary relation on $W$ and $s_1, s_2$ are functions on $W$, 
\item[(K1)] $R$ is a preorder, that is $R$ is reflexive and transitive,
\item[(K2)] $s_j(s_i(w)) = s_j(w)$, for all $i,j = 1, 2$,
\item[(K3)] $R(s_1(w), w)$,
\item[(K4)] $R(w, s_2(w))$,
\item[(K5)] $R(w, w')$ implies  $R(s_i(w), s_i(w'))$ and $R(s_i(w'), s_i(w))$, for $i = 1, 2$,
\item[(K6)] If $w \in W$ then there are $i \in \{1, 2\}$ and $w' \in W$ such that $w = s_i(w')$.
\end{itemize}
\end{definition}

\smallskip

\begin{definition}\label{def:HT-model}
A \textbf{$\boldsymbol{HT}$-model} based on a HT-frame $K$ is a system
$M = (K, m)$ such that $m : \textbf{VarProp} \rightarrow {\cal P}(W)$ is a meaning function that assigns subsets of states to propositional variables, and satisfies the atomic heredity condition:
\[
\text {(her \ at)} \qquad R(w, w') \ \text{and} \  w \in m(p) \  \text{imply} \  w' \in m(p).
\]
\end{definition}

We say that in a $HT$-model $M$ \textbf{a state $\boldsymbol w$ satisfies a formula $\boldsymbol \alpha$} (denoted $M, w \ sat \ \alpha$) whenever the following conditions are satisfied:

\begin{tabular}{lll}
$M, w$ \ sat \ $p$ & iff & $w \in m(p),$ \quad for $p \in \textbf{VarProp}$,  \\
$M, w \ sat \ \alpha \wedge \beta$ & iff & $M, w \ sat \ \alpha$ \ and \ $M, w \ sat \ \beta$,\\
$M, w \ sat \ \alpha \vee \beta$ & iff & $M, w \ sat \ \alpha$ \ or \ $M, w \ sat \ \beta$,\\
$M, w \ sat \ \alpha \Rightarrow \beta$ & iff & for all $w'$, if $R(w, w')$ and $M, w' \ sat \ \alpha$ then, $M, w' \ sat \ \beta$,\\
$M, w \ sat \ \neg \alpha$ & iff & for all $w'$, if $R(w, w')$ then, $not \ M, w' \ sat \ \alpha$,\\
$M, w \ sat \ S_i \alpha$ & iff & $M, s_i(w) \ sat \ \alpha$.
\end{tabular}

\

Given a $HT$-model $M$, we extend the meaning function $m$ to all formulas:
\[
m(\alpha) = \{w \in W : M, w \ sat \ \alpha \}
\]

\medskip

A formula $\alpha$ is \textbf{true in a $\boldsymbol{HT}$-model} $M = (K, m)$ (denoted $\boldsymbol{M \ sat \ \alpha}$) iff $M, w \ sat \ \alpha$, for every $w \in W$ (i.e.\ $m(\alpha) = W$), $\alpha$ is \textbf{true in a $\boldsymbol{HT}$-frame} $K$ iff it is true in every $HT$-model based on $K$, and $\alpha$ is \textbf{$\boldsymbol{HT}$-valid} (denoted $\models_{Rel} \alpha$) iff it is true in every $HT$-frame.

A formula $\alpha$ is a \textbf{relational $\boldsymbol{HT}$-consequence} of a set of formulas $\Gamma$ in a $HT$-model $M = (K, m)$ (denoted by $\Gamma \models_{M} \alpha$) iff whenever all the formulas from $\Gamma$ are true in $M$, we have $\alpha$ is true in $M$; 
and $\alpha$ is a \textbf{relational $HT$-consequence} of a set of formulas $\Gamma$ (denoted by $\Gamma \models_{Rel} \alpha$) iff for every $\boldsymbol{HT}$-model $M$, $\Gamma \models_{M} \alpha$).

\smallskip

\begin{proposition} \label{prop: her} 
For every $HT$-model $M = (K, m)$ and for every 
formula $\alpha$ the following heredity condition holds:
\[
{(her) \qquad if\ } R(w, w') \ \text{and} \ M, w \ sat \ \alpha, 
\ \text{then}\ M, w' \ sat \ \alpha.
\]
\end{proposition}
\begin{proof}
The proof is by induction with respect to complexity of $\alpha$. By the way of example we show (her) for formulas of the form $S_i \alpha$. Let $R(w, w')$ and 
$M, w \ sat \ S_i \alpha$, hence by $(K5)$ we have $R(s_i(w), s_i(w'))$ and by Definition \ref{def:HT-model} we deduce $M, s_i(w) \ sat \ \alpha$. From the inductive hypothesis we obtain $M, s_i(w') \ sat \ \alpha$, i.e.\ $M, w' \ sat \ S_i \alpha$.
\end{proof}

\smallskip

\begin{proposition} In every $HT$-frame $K = (W, R, s_1, s_2)$, for every 
$w \in W$, there is $i \in \{1, 2\}$ such that $w = s_i(w)$, i.e.\ each $w$ is 
a fixed point of a function $s_i$.
\end{proposition}
\begin{proof}
Let $w \in W$. By $(K6)$ there are $i \in \{1, 2\}$ and $w' \in W$ such that $w = s_i(w')$. Hence $s_i(w) = s_i (s_i(w')) = s_i(w') = w$.
\end{proof}

\section{\hspace{-2.6ex}.{\hspace{1.ex}}Equivalence of algebraic and relational model validity}

First let us suppose that we have a $HT$-model $M = (W, R, s_1, s_2, m)$. 
We will define an \textbf{algebraic model} $(A, 0, 1, \wedge, \vee, \Rightarrow, \neg, 
S_1, S_2, h)$ such that for any formula $\alpha$:
\[
h(\alpha) = 1 \quad \text{if and only if} \quad  \ M \  sat \ \alpha
\]

A subset $X \subseteq W$ will be called \textbf{$R$-closed} if whenever $w \in X$ and $R(w, w')$, then $w' \in X$.

\

Let $\textbf{RC}$ be the collection, ordered by inclusion ($\subseteq$), of all $R$-closed subsets of $W$:
\[
\textbf{RC} = \{ X \subseteq W : \text{\ X is R-closed\}}.
\]

We can consider on $RC$ the operations of intersection $\cap$ and union $\cup$. The system $(RC, \emptyset, W, \cap, \cup)$ is a distributive lattice with zero and unit.

Also, if $X, Y \in RC$, let us consider the sets:
\begin{align*}
S_i X &= \{w \in W : s_i(w) \in X \}\ =\  s_{i}^{-1}(X),\ {\text \ for\ i = 1, 2,}\\
C X &= {\cal C}_{W} S_1 X,  
\end{align*}
where ${\cal C}_{W}$ is the ordinary set complementation.

If $X, Y$ are $R$-closed then $S_1 X$, $S_2 X$ and $C X$ are $R$-closed. In fact, assume $w' \in S_i X$ and $R(w', w'')$. By $(K5)$ in Definition \ref{def:HT-frame} we obtain $R(s_i(w'), s_i(w''))$. 
Since $X$ is $R$-closed and $s_i(w') \in X$ we deduce $s_i(w'') \in X$, i.e.\ $w'' \in S_i X$. 

To prove ${\cal C}_{W}S_1 X \in RC$ assume $w' \in {\cal C}_{W}S_1 X$ and $R(w', w'')$. 
From $(K3)$ we have $R(s_1(w'), w')$. By transitivity of $R$ we get $R(s_1(w'), w'')$.
By $(K5)$, we have $R(s_1(w''), s_1s_1(w'))$ and by $(K2)$ we get $R(s_1(w''), s_1(w'))$. 
Since $X$ is $R$-closed and $s_1(w') \not \in X$ it follows that $s_1(w'') 
\not \in X$, hence $w'' \not \in S_1 X$, i.e.\ $w'' \in {\cal C}_{W}S_1 X$.

Moreover $S_1 \emptyset = \emptyset$ and $S_1 W = W$. 

\smallskip

\begin{proposition} The system 
$(RC, \emptyset, W, \cap, \cup, C, S_1, S_2)$ is a $T$-structure.
\end{proposition}
\begin{proof}
We show that the operations defined above fulfill the properties $(T1)-(T7)$ 
indicated in Definition \ref{def:alg2}. 

In fact, $(T1)$ and $(T5)$ have been indicated above; $(T2)$ follows at once from properties of the inverse image; $(T3)$ is a consequence of definitions; $(T4)$ is a consequence of $(K2)$ and a property of the inverse image.

To prove $(T6)$ suppose $S_i X = S_i Y$, for all $i = 1, 2$. Let $w \in X$; by $(K6)$ there is $i \in \{1, 2\}$ and 
$w' \in W$ such that $w = s_i(w') \in X$. It follows that $w' \in S_i X = 
S_i Y$, that is $s_i(w') = w \in Y$ and thus $X \subseteq Y$. The proof of the other half is similar. 

Finally, to prove $(T7)$ let $w' \in S_1 X$, that is $s_1(w') \in X$. By $(K3)$ 
we have $R(s_1(w'), w')$ and since $X$ is $R$-closed we obtain $w' \in X$. Suppose now $w \in X$. 
By $(K4)$ we have $R(w, s_2(w))$ so $s_2(w) \in X$ and $w \in S_2 X$.

The proof of the proposition is now complete.
\end{proof}

\smallskip

\begin{remark}
Taking into account the equivalence between Definitions \ref{def:alg2} and \ref{def:alg3} we find in particular the well known result that the system $(RC, \emptyset, W, \cap, \cup, \Rightarrow, \neg)$ is a Heyting 
algebra (\cite{Fit69}, page 24). 

For sets $X, Y \in RC$, the set $X \Rightarrow Y$ is given by the 
equation (1) in Theorem \ref{theo:impl}. That is:
\begin{align*}
X \Rightarrow Y & =  Y \cup \bigcap_{k=1}^{2} ({\cal C}_{W}S_1 S_k X \cup S_k Y)\ =\  Y \cup \bigcap_{k=1}^{2} (s_{k}^{-1}({\cal C}_{W} X) \cup s_{k}^{-1}(Y)\\
& =  Y \cup s_1^{-1} ({\cal C}_{W} X) \cup s_{1}^{-1}(Y)\ =\  Y \cup s_{1}^{-1}({\cal C}_{W} X)
\end{align*}

Thus $Z = S_1 ({\cal C}_{W} X) \cup Y$ is the largest $R$-closed subset such that $X \cap Z \subseteq Y$.
\end{remark}

\smallskip

We define $h: For \rightarrow {\cal P}(W)$ by
\[
h(\alpha) = \{ w \in W : M, w \text {\ sat\ } \alpha\}
\]

Let $w \in h(\alpha)$ and $R(w, w')$. By {\it (her)} (Proposition \ref{prop: her}), 
$M, w' \ sat\ \alpha$, i.e.\ $w' \in h(\alpha)$. Thus $h(\alpha)$ is $R$-closed.

From a result in (\cite{Fit69}, page 24), we know that $h$ is a Heyting 
homomorphism. Moreover we have the equality  $h(S_i \alpha) = S_i h(\alpha).$
This fact is a consequence of the following equivalent conditions:
\[
\begin{array}{lllll}
w \in h(S_i a) & \Longleftrightarrow & M, w {\ sat\ } S_i \alpha 
& \Longleftrightarrow & M, s_i(w) {\ sat\ } \alpha \\
& \Longleftrightarrow & s_i(w) \in h(\alpha) 
& \Longleftrightarrow & w \in S_i h(\alpha)
\end{array}
\]

\noindent Thus $(RC, \emptyset, W, \cap, \cup, \Rightarrow, \neg, S_1, S_2, h)$ is an \textbf{algebraic model}. 

\medskip

Concerning the validity of a formula, we have the desired equivalence:
\[
h(\alpha) = W (\in RC) \quad \text{if and only if}\quad m(\alpha) = W
\]

\

\textbf{Conversely}, suppose we have an algebraic model $(A, h)$. We will define a \textbf{$\boldsymbol{HT}$-model} $M = (W, R, s_1, s_2, m)$ such that for any formula $\alpha$:
\[
M\ sat\ \alpha \quad \text{if and only if}\quad h(\alpha) = 1
\]

\smallskip

Let $W$ be the class of all \textbf{prime filters} in $A$. Let $R$ be the inclusion relation $\subseteq$ and $s_i : W \rightarrow W$ be the maps defined as follows, for $i=1,2$ and $P \in W$: 
\[
s_i(P) = \{x \in A : S_i x \in P \}. 
\]
This set is a prime filter.

If $p \in VarProp$ and $P$ is a prime filter, we define 
\[
M, P \ sat \ p \quad \text{if and only if}\quad h(p) \in P.
\]

\smallskip

\begin{proposition} \label{prop:frame}
If $A$ is a $HT$-algebra, the system $K = (W, R, s_1, s_2)$ defined above is a $HT$-frame.
\end{proposition}
\begin{proof}
$(K0)$ follows from the definition of $K$. The relation $\subseteq$ satisfies $(K1)$; $(K2)$ is a consequence of $(HT5)$ and a property of the inverse image; $(K3)$ and $(K4)$ are consequence of $(HT6)$, $(T10)$ and a property of prime filters.
 
To prove $(K5)$ suppose $P, Q \in W$ and $P \subseteq Q$. Let $x \in 
s_i(P)$ then $S_i x \in P \subseteq Q$, hence $x \in s_i(Q)$. In addition, let $x \in s_i(Q)$, i.e.\ $S_i x \in Q$. If $S_i x \not \in P$ then 
$\neg S_i x \in P \subseteq Q$ and $S_i x \wedge \neg S_i x = 0 \in Q$, 
which is impossible; hence $S_i x \in P$, i.e.\ $x \in s_i(P)$.  

Finally, to prove $(K6)$, assume $P \in W$. By theorem 5.10 in 
(\cite{Itu01}, p.149) there exists a unique ultrafilter $P' \ (= P \cap B(A))$ in 
$B(A)$ and an integer $i \in \{1, 2\}$ such that $ P = P_i' = \{x \in A : 
S_i x \in P\} = s_i(P)$.
\end{proof}

\smallskip

\begin{proposition}  \label{prop:model}
For formulas $\alpha, \beta$, and prime filters $P, Q$ we have:
\begin{itemize}
\item[(1)] If $M, P \ sat \ \alpha$ and $P \subseteq Q$ then $M, Q \ sat \ \alpha$
\item[(2)] $M, P \ sat \ (\alpha \wedge \beta)$ iff $M, P \ sat \ \alpha$ and $M, P \ sat \ \beta$
\item[(3)] $M, P \ sat \ (\alpha \vee \beta)$ iff $M, P \ sat \ \alpha$ or $M, P \ sat \ \beta$
\item[(4)] $M, P \ sat \ (\alpha \Rightarrow \beta)$ iff for every $Q \in W$, if $P \subseteq Q$, and $M, P \ sat \ \alpha$ then $M, Q \ sat \ \beta$
\item[(5)] $M, P \ sat \ \neg \alpha$ iff for every $Q \in W$ such that $P \subseteq Q$ then $not\ M, Q \ sat \ \alpha$
\item[(6)] $M, P \ sat \ S_i \alpha$ iff $M, s_i(P) \ sat \ \alpha$
\end{itemize}
\end{proposition}
\begin{proof}
We show, for example, the reverse implication of $(4)$ and the statement $(6)$.
Suppose $not\  M, P \ sat \ (\alpha \Rightarrow \beta)$, i.e.\ $h(\alpha \Rightarrow \beta) = h(\alpha) \Rightarrow h(\beta) \not \in P$. Since $h(\beta) \subseteq h(\alpha) \Rightarrow h(\beta)$ we deduce that $h(\beta) \not \in P$.
Let $F(P, h(\alpha))$ be the filter generated by $P$ and $h(\alpha)$. This filter is proper because for example $h(\beta) \not \in F(P, h(\alpha))$. In fact, if $h(\beta) \in F(P, h(\alpha))$, then there would be some $p \in P$ such that $p \wedge h(\alpha) \leq h(\beta)$ which is equivalent to $p \leq h(\alpha) \Rightarrow h(\beta) \in P$, a contradiction. Since $A$ is a distributive lattice then there is a prime filter $Q$ such that $F(P, h(\alpha)) \subseteq Q$ and $h(\beta) \not \in Q$. By construction, $P \subseteq Q$ and $h(\alpha) \in Q$. That is $M, Q \ sat \ \alpha$. 
Hence, by hypothesis, $M, Q \ sat \ \beta$, i.e.\ $h(\beta) \in Q$, a contradiction. 

Statement $(6)$ is a consequence of the following equivalent conditions:
\[
\begin{array}{lllll}
M, P \ sat \ S_i \alpha & \Longleftrightarrow & h(S_i \alpha) = S_i h(\alpha) \in P \\
& \Longleftrightarrow & h(\alpha) \in s_i(P) \\ 
& \Longleftrightarrow & M, s_i(P) \ sat \ \alpha 
\end{array}
\]
\end{proof}

We define $m: For \rightarrow {\cal P}(W)$ such that $m(\alpha) =  \{P \in W : \ M, P \ sat \ \alpha \}$. Thus the obtained system $M = (K, m)$ is a \textbf{$\boldsymbol{HT}$-model}

\medskip.

Concerning the validity of a formula $\alpha$, we have:
\[
\begin{array}{lll}
m(\alpha) =  \{P \in W : \ M, P \ sat \ \alpha \} = W 
& \Longleftrightarrow & h(\alpha) \in P,\ \text{for every}\ P \in W \\
& \Longleftrightarrow & h(\alpha) \in \bigcap_{P \in W} P \\
& \Longleftrightarrow & h(\alpha) = 1
\end{array}
\]

\medskip

Summing up the above results we will provide the expected result, which is useful in applications:
\begin{theorem} 
A formula $\alpha$ is a relational consequence of a set of formulas $\Gamma$ if and only if $\alpha$ is an algebraic consequence of $\Gamma$.
\end{theorem}
\begin{proof}
The statement can be formally written in the following way: 
\[
\Gamma \models_{Rel} \alpha \quad \text{if and only if }\quad \Gamma 
\models_{Alg} \alpha
\] 

$(\rightarrow)$ Assume $\Gamma \models_{Rel} \alpha$. If $\Gamma \not \models_{Alg} \alpha$, there would be an algebraic model $(A, h)$ such that $h(\gamma) = 1$, for all $\gamma \in \Gamma$ but $h(\alpha) \neq 1$. 
Let $M = (K, m)$ be the $HT$-model for $\Gamma$ defined in Propositions \ref{prop:frame} and \ref{prop:model}.
Since $h(\alpha) \neq 1$ there is a prime filter $P$ in $A$ such that $h(\alpha) \not \in P$, i.e.\ $m(\alpha) = \{P \in W : h(\alpha) \in P \} \neq W$, a contradiction.

$(\leftarrow)$ Conversely, suppose $\Gamma \models_{Alg} \alpha$. If $\Gamma \not \models_{Rel} \alpha$, there would be a $HT$-model $M = (K, m)$ for $\Gamma$ such that $M \ sat \ \gamma$, for all $\gamma \in \Gamma$ but $not \ M \ sat \ \alpha$. 

Let $(RC, \emptyset, W, \cap, \cup, C, S_1, S_2, h)$ be the algebraic model of $R$-closed subsets of $W$, where $h: For \rightarrow RC$ is the homomorphism: $h(\alpha) = \{w \in W: M, w \ sat\ \alpha\}$, for $\alpha \in For$.  
We have $h(\gamma) = W$ for every $\gamma \in \Gamma$ but $h(\alpha) \neq W$, a contradiction.
\end{proof}

\medskip

In particular if $\Gamma$ is empty we can conclude the following fact:
\begin{theorem}
A formula $\alpha$ is valid in every relational model if and only if $\alpha$ is algebraically valid.
\end{theorem}

\section{\hspace{-2.6ex}.{\hspace{1.ex}}A finite algebraic model}

In this section we show that there is an effective method whereby, 
for any given formula $\alpha$, it can be determined in a finite number of steps whether or not $\alpha$ is an algebraic consequence of a finite set of formulas $\Gamma$. Thus, the formalised propositional system introduced in Section 4 is decidable. 

This result is a consequence of the following theorem.
\begin{theorem} A formula $\alpha$ is an algebraic consequence of a finite set of formulas $\Gamma$ if and only if we have $\Gamma \models_{BT} \alpha$, for every algebraic model $(\boldsymbol{BT}, h)$ based on the finite $HT$-algebra $\boldsymbol{BT}$.
\end{theorem}
\begin{proof}
The statement can be formally written in the following way, for a finite set $\Gamma$ of formulas: 
\[
\Gamma \models_{Alg} \alpha \quad \text{if and only if} \quad \Gamma 
\models_{BT} \alpha
\]

$(\rightarrow)$ Assume $\Gamma \models_{Alg} a$. Thus, for every algebraic model $(A, h)$, if $h(\gamma) = 1$ for all $\gamma \in \Gamma$, then $h(a) = 1$. In particular in the case $A = \textbf{BT}$.

$(\leftarrow)$ Conversely, suppose $\Gamma \models_{BT} \alpha$ and let $(A, h)$ be any algebraic model such that $h(\gamma) = 1$, for all $\gamma \in \Gamma$.

If $\alpha$ is not true in $(A,h)$, there would be a minimal prime filter $P$ in $A$ (see \cite{Itu01}) such that $h(\gamma) \in P$, for all $\gamma \in \Gamma$ but $h(\alpha) \not \in P$. 
Let $f : A \rightarrow BT$ be the canonical homomorphism defined -via the quotient algebra $A/P$, isomorphic to a subalgebra de $BT$- as in the proof of Proposition 6.2 in (\cite{Itu01}, p.152).

The composition $g = f$\ o\ $h : For \rightarrow BT$ is a homomorphism which satisfies $g(\gamma) = 1$, for all $\gamma \in \Gamma$ and $g(\alpha) \neq 1$. This means that $\Gamma \not \models_{BT} \alpha$, a contradiction.
\end{proof}

\

Finally we point out another link between finite models. The system $K^0 = (W^0, R^0, s_1^0, s_2^0)$ related to the finite $HT$- algebra $BT$ can be defined in the following way:
\begin{itemize}
\item[-] $W^0 = T$, 
\item[-] $R^0 = \{(w, w'): w, w' \in W^0 \ \text{and}\ w \leq w' \}$, 
i.e.\ $R^0$ is the order on $T$,
\item[-] if $t \in T$ and $w \in W^0$ then $s_i^0(w) = t_{i}$, for $i = 1, 2$.
\end{itemize}

\begin{proposition}
The system $K^0 = (W^0, R^0, s_1^0, s_2^0)$ satisfies the properties $(K0)-(K6)$, that is $K^0$ is a $HT$-frame.
\end{proposition}
\begin{proof}
The proof is straightforward.
\end{proof}

\medskip

We note that, given the $HT$-frame $K^0$, the collection of $R$-closed sets of $W^0$ is $\{\emptyset, F(t_{2}), F(t_{1})\}$. 

As in Section 5, we can construct the $T$-structure $(RC, \emptyset, W^0, \cap, \cup, C, S_{1}, S_{2})$, which is isomorphic to the basic $T$-structure $BT$.

\

In view of the results above, we conclude the paper with the following statement.

\begin{proposition}
Let $K^0$ be the $HT$-frame defined above, $\Gamma$ a finite set of formulas, and $\alpha$ a formula. It follows that:

\begin{tabular}{lll}
$\Gamma \models_{BT} \alpha$ &iff& for every $HT$-model $M^0 = (K^0, m)$, we have: $\Gamma \models_{M^0} \alpha$.
\end{tabular}
\end{proposition}

\

\noindent \textbf{Acknowledgement}

\noindent The author would like to thank the anonymous expert who read the paper and has provided useful advice and corrections.

\centerline{--------------------}

\begin{thebibliography}{99}
\bibitem{BFGR} Boicescu, V., Filipoiu, A., Georgescu G., Rudeanu S.,
\textit{{\L}ukasiewicz-Moisil Algebras}, Annals of Discrete Mathematics \textbf{49}, North-Holland, 1991, 583 pages.

\bibitem{Fit69} Fitting, M.C., \textit{Intuitionistic Logic Model Theory and Forcing}, North-Holland, Amsterdam, 1969, 191 pages. 

\bibitem{Itu82} Iturrioz, L., \textit{Modal Operators on Symmetrical Heyting algebras}, Universal Algebra and Applications, Banach Center Publications \textbf{9}, Traczyk T. (ed.), PWN-Polish Scientific Publishers, 1982, 289--303.

\bibitem{Itu01} Iturrioz, L., \textit{Algebraic Structures Based on a Chain of Two 
Agents}, in Multiple-Valued Logic - An International Journal, `Grigore C. Moisil memorial issue', Rudeanu S., Iorgulescu A., Georgescu G., Ionita C. (eds.), \textbf{6} no 1-2 (2001), 137--155.

\bibitem{IO-96} Iturrioz L., Or{\l}owska E., \textit{A Kripke-style and relational semantics for logics based on {\L}ukasiewicz algebras}, in {\L}ukasiewicz in Dublin, an International Conference on the Work of Jan {\L}ukasiewicz, Dublin, Ireland, 7--10 July 1996. The paper is published in: J. of Multi-Valued Logic \& Soft Computing \textbf{12} (2006), 131--147.

\bibitem{Itu07} Iturrioz, L., \textit{Two representation theorems of three-valued structures by means of binary relations},  arXiv:0710.1007v1 [cs.DM] 4 Oct 2007, 11 pages.
 
\bibitem{Moi40} Moisil, Gr.C., \textit{Recherches sur les logiques non chrysippiennes}, Annals Sci. Univ. Jassy \textbf{26} (1940), 431--466.

\bibitem{Moi72} Moisil, Gr.C., \textit{Essais sur les logiques non chrysippiennes}, Editions de l'Acad. Rep. Soc. de Roumanie, Bucarest, 1972, 820 pages.

\bibitem{Mont-L64} Monteiro, L., \textit{Alg\`ebre du calcul propositionnel trivalent de Heyting}, Fundamenta Mathematicae \textbf{74} (1972), 99--109.

\bibitem{Ras74} Rasiowa, H., \textit{An Algebraic Approach to Non-Classical Logics}, Studies in Logic and the Foundations of Mathematics \textbf{78}, North-Holland, Amsterdam, 1974, 403 pages.

\bibitem{IvoT} Thomas, I., \textit{Finite limitations on Dummett's \textbf{LC}}, Notre Dame J. of Formal Logic \textbf{3} (1962), 170--174.
\end{thebibliography}
\end{document}